\documentclass[pdftex,12pt]{article}
\usepackage[pdftex]{graphicx}
\usepackage[normalem]{ulem}
\usepackage{amsmath,amsthm,amsfonts,
amsbsy,amssymb,upref,enumerate,bigstrut,
color,mathtools,mathrsfs,float,bm} 
\usepackage{lipsum}
\usepackage{booktabs}
\usepackage{epsfig}
\usepackage{fullpage}
\usepackage{graphicx,psfrag,epsf}
\usepackage{indentfirst}
\usepackage{natbib}
\usepackage{setspace}
\usepackage{subfigure}
\usepackage{array}
\usepackage{times}
\usepackage{verbatim}
\usepackage{fancyhdr}
\usepackage{setspace} 
\usepackage{natbib}
\usepackage{multirow}
\usepackage{url}
\usepackage{hyperref}
\hypersetup{colorlinks = true,
	linkcolor = black,
	urlcolor  = black,
	citecolor = black}
\usepackage[inline]{enumitem}
\usepackage{mathtools}
\usepackage{multicol}

\usepackage{float}

\usepackage[left=0.8in,top=1in,right=0.8in]{geometry}

\newtheorem{theorem}{Theorem}[section]

\newtheorem{proposition}[theorem]{Proposition}

\numberwithin{equation}{section} 

\makeatletter
\def\@seccntformat#1{\@ifundefined{#1@cntformat}%
	{\csname the#1\endcsname\quad}
	{\csname #1@cntformat\endcsname}
}
\makeatother

\newif\ifShowComments
\ShowCommentstrue
\def\strutdepth{\dp\strutbox}
\def\druk#1{\strut\vadjust{\kern-\strutdepth
        {\vtop to \strutdepth{%
                \baselineskip\strutdepth\vss
                        \llap{\hbox{#1}\quad}\null}}}}

\newcommand{\myquad}[1][1]{\hspace*{#1em}\ignorespaces}
\DeclareMathOperator{\E}{\mathrm{E}}

\title{ \bf 
On a length-biased Birnbaum-Saunders regression model applied to meteorological data}
\author{Kessys L. P. Oliveira,\,
Bruno S. Castro,\,
Helton Saulo\, and\,  
Roberto Vila
\\ \small Department of Statistics, Universidade de Bras\'ilia, Bras\'ilia, Brazil 
}

\allowdisplaybreaks

\begin{document}
\maketitle

\begin{abstract}
The length-biased Birnbaum-Saunders distribution is both useful and practical for environmental sciences. In this paper, we initially derive some new properties for the length-biased Birnbaum-Saunders distribution, showing that one of its parameters is the mode and that it is bimodal. We then introduce a new regression model based on this distribution. We implement use the maximum likelihood method for parameter estimation, approach interval estimation and consider three types of residuals. An elaborate Monte Carlo study is carried out for evaluating the performance of the likelihood-based estimates, the confidence intervals and the empirical distribution of the residuals. Finally, we illustrate the proposed regression model with the use of a real meteorological data set.  \\
{\small {\bfseries Keywords.} Length-biased model; Mode regression; Bimodality; Monte Carlo simulation; Meteorological data.}
\end{abstract}



\smallskip 

\onehalfspacing

\section{Introduction}\label{sec:introduction}

Birnbaum-Saunders (BS) regression models have been widely used in the literature; see the literature review by \cite{bk:19}. Recently, \cite{ddlms:20} performed a comparison of three existing regression approaches, studied by \cite{rn:91}, \cite{lscb:14} and \cite{bz:15}, to deal with the modeling of asymmetric data following the BS distribution. Other recent studies involving BS regression models can be seen in \cite{slgs:2020a,slgs:2020b} and \cite{lsgs:2020c}. 

Length-biased distributions are special cases of weighted distributions; see \cite{sa:01}. In the area of environmental sciences, as noted by \cite{p:06}, the use of weighted distributions is more adequade, since observations from this area fall in the nonexperimental, nonreplicated, and nonrandom categories. Weighted distributions take into account these caracteristichcs by providing probability-adjusted models that consider the method of ascertainment; see \cite{p:06}.

In the context of BS models, \cite{lsa:09} proposed a length-biased version of the BS (LBS) distribution. The authors provided moments and some properties of this distribution, which was illustrated with real data related to water quality. Nevertheless, no regression model based on the LBS distribution has been proposed in the literature. Therefore, the primary objective of this paper is to propose a regression model based on the LBS distribution. The secondary objectives are: (i) to investigate some properties of the LBS distribution, such as bimodality and mode; (ii) to obtain point and interval estimates of the model parameters; (iii) to carry out Monte Carlo simulations to evaluate the performance of the estimates; and (iv) to discuss a real data application of the proposed methodology.

The rest of this paper proceeds as follows. In Section \ref{sec:2}, 
we describe briefly the LBS distribution proposed by \cite{lsa:09} and present some novel properties of this model. In Section \ref{sec:3}, we formulate the regression model based on the LBS distribution, and then detail the associated point and interval estimation and residual analysis. In Section \ref{sec:4}, we carry out Monte Carlo simulations, and an illustration with a meteorological data is done in Section \ref{sec:5}. Finally, in Section \ref{sec:6}, we discuss conclusions and some possible future research in this topic.

\section{The LBS distribution and some novel properties}\label{sec:2}

In this section, we briefly describes the LBS distribution. Then, some novel results on bimodality and monotonicity of the hazard rate (HR) of the LBS distribution are obtained.

\subsection{The LBS distribution}\label{sec:2.1}

Let $Y$ be a positive random variable with probability density function (PDF) $f_{Y}$. Then, the length-biased version of $Y$, denoted by $T$, has PDF
\begin{equation}\label{LB-density}
f_{T}(t) = \frac{tf_{Y}(t)}{\E(Y)}, \quad t>0, 
\end{equation}
provided the expectation $\E(Y)$ exists. In the case of the LBS distribution, the random variable $Y$ follows a BS distribution \citep{bs:69} with shape parameter $\alpha$ and scale parameter $\theta$, denoted by $Y \sim \textrm{BS}(\alpha, \theta)$, with $\E(Y) ={\theta}(\alpha^2 + 2)/2$. Thus, $T$ is a random variable following a LBS distribution with PDF given by
\begin{equation}
\label{eq:fdp}
f_{T}(t)=f_{T}(t;\alpha, \theta)
=
\frac{1}{\sqrt{2\pi}\alpha \theta(\alpha^2 + 2)} \left[ \left(\frac{t}{\theta} \right)^\frac{1}{2} + \left(\frac{\theta}{t} \right)^\frac{1}{2} \right] \exp \left[-\frac{1}{2\alpha^2} \left(\frac{t}{\theta}+\frac{\theta}{t}-2 \right) \right],
\end{equation}
with $t>0$, $\alpha>0$ and $\theta>0$. In this case, we write $T \sim \textrm{LBS}(\alpha, \theta)$. According to \cite{lsa:09}, the parameter $\theta$ in \eqref{eq:fdp} relates only to the scale, while the parameter $\alpha$ controls asymmetry and kurtosis of the distribution. In Subsection \ref{sec:2.2}, we prove that the parameter $\theta$ is the mode when $\alpha\leqslant 2$ and the distribution is bimodal when $\alpha> 2$; Figure \ref{figpdf:1} displays different shapes of the LBS PDF for different combinations of parameters. 

\begin{figure}[h!]
\vspace{-0.25cm}
\centering
{\includegraphics[scale=0.57]{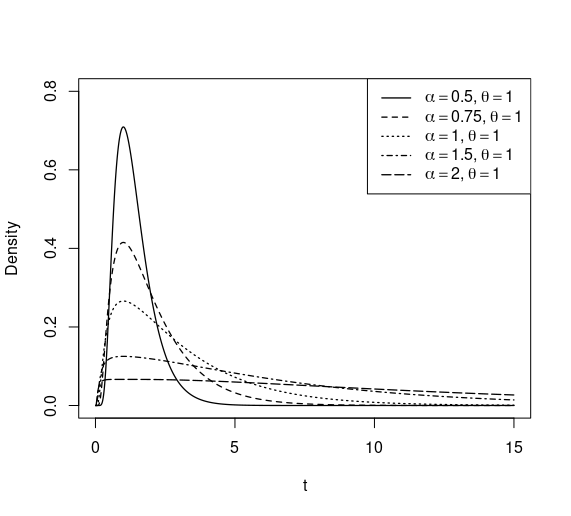}}
{\includegraphics[scale=0.57]{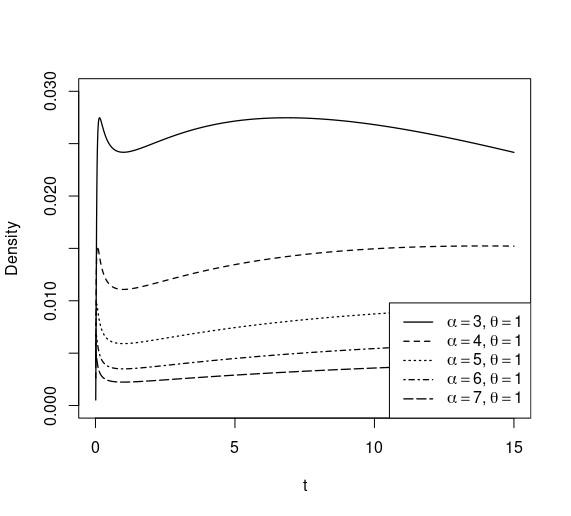}}\\
\vspace{-0.2cm}
\caption{LBS PDFs for some parameter values.}\label{figpdf:1}
\end{figure}

Let $T \sim \textrm{LBS}(\alpha, \theta)$; we then readily have the following properties \citep{lsa:09}: (P1) $cT \sim \textrm{LBS}(\alpha, c\theta)$, with $c>0$; (P2) $\E(T) = \theta\left(\dfrac{2+4\alpha^2+3\alpha^4}{2+\alpha^2}\right)$; (P3) $\textrm{Var}(T) = \theta^2\alpha^2\left[\dfrac{4+17\alpha^2+24\alpha^4+6\alpha^6}{(2+\alpha^2)^2}\right]$; (P4) $\E\big[T^{-(r+1)}\big] = \dfrac{\E(Y^r)}{\theta^{2r}\E(Y)}$, where $Y \sim \textrm{BS}(\alpha,\theta)$; and (P5) $U=\frac{1}{\alpha^2}\left(\frac{T}{\beta}+\frac{\beta}{T}-2 \right)$ has PDF $f_U(u)=\pi f_{U_1}(u)+(1-\pi) f_{U_2}(u)$, where $\pi=2/(\alpha^2+2)$, $U_1\sim \textrm{Gamma}(1/2,2)$ and $U_2\sim \textrm{Gamma}(3/2,2)$.

\subsection{Novel properties}\label{sec:2.2}

\subsubsection{Bimodality properties}\label{sec:2.2.1}
In order to state and prove the main result of this subsection (Theorem \ref{Unimodality-bimodality}),
we define the following function
\begin{align}\label{definition-a}
a(t)=\dfrac{1}{\alpha}\left[\sqrt{t\over \theta}-\sqrt{\theta\over t}\right], \quad t>0.
\end{align}
Notice that the $n$-th derivative of $a(t)$, denoted by $a^{(n)}(t)$, satisfies $a^{(n)}(t)> 0$ (or $< 0$) for $n$
odd (or $n$ even), where $n\geqslant 1$. Some special cases of these derivatives, when $n=1,2,3$, are of the following form:
\begin{align}\label{derivatives-a}
\hspace*{-0.3cm}
\begin{array}{lllll}
\displaystyle
a'(t)=\dfrac{1}{2\alpha t}\left[\sqrt{t\over \theta}+\sqrt{\theta\over t}\right];
&
\displaystyle
a''(t)=-\dfrac{1}{4\alpha t^2}\left[\sqrt{t\over \theta}+3\sqrt{\theta\over t}\right];
&
\displaystyle
a'''(t)=\dfrac{3}{8\alpha t^3}\left[\sqrt{t\over \theta}+5\sqrt{\theta\over t}\right]. 
\end{array}
\end{align}
\begin{proposition}[Modes]\label{Mode}
A mode $t$ of the LBS distribution is obtained by resolving the following cubic equation 
\begin{align*}
t^3 - \theta (\alpha^2-1) t^2 + \theta^2 (\alpha^2-1) t -\theta^3 = 0.
\end{align*}
\end{proposition}
\begin{proof}
Let  $T \sim \textrm{LBS}(\alpha, \theta)$.
A simple computation shows that the $n$th derivative of $f_T$ is given by 
\begin{align}\label{n-derivative}
f^{(n)}_{T}(t) 
=
\dfrac{1}{\E(Y)} \,
\left[
n f^{(n-1)}_{Y}(t) + t f^{(n)}_{Y}(t)
\right], \quad n\geqslant 1,
\end{align}
where $Y \sim \textrm{BS}(\alpha, \theta)$ and we are denoting $f^{(0)}_{Y}(t)=f_{Y}(t)$. Since
\begin{align}
f'_{Y}(t) = \phi\big[a(t)\big]
\left\{a''(t)-a(t)[a'(t)]^2\right\}, \label{derivative-f}
\end{align}
by combining \eqref{n-derivative} and \eqref{derivative-f} we have
\begin{align*}
f'_{T}(t) 
=
\dfrac{\phi\big[a(t)\big]}{\E(Y)} \,
\left\{
a'(t) + t \left\{a''(t)-a(t)[a'(t)]^2\right\}
\right\}.
\end{align*}
By using \eqref{definition-a}, \eqref{derivatives-a} and simple algebraic manipulations we get that $f'_{T}(t) =0$ if and only if 
\begin{align*}
t^3 - \theta (\alpha^2-1) t^2 - \theta^2 (1-\alpha^2) t -\theta^3 = 0.
\end{align*}
Therefore, a mode $t$ of the LBS distribution must satisfy the above equation.

\end{proof}

The next theorem reveals that the parameter $\alpha$ of the LBS distribution, in addition to controlling asymmetry, controls the uni- or bimodal shape of the distribution regardless of the parameter $\theta$.
\begin{theorem}[Unimodality and Bimodality]
\label{Unimodality-bimodality}
The PDF of the  LBS distribution \eqref{eq:fdp} has the following shapes:
\begin{itemize}
\item[1)] It is unimodal as $\alpha\leqslant 2$, with mode $t_0=\theta$;
\item[2)] It is bimodal as $\alpha> 2$, with modes 
\begin{align}\label{modes-t}
t_{\pm}=
{\theta\over 2}\,
\Big[(\alpha^2-2)\pm\alpha\sqrt{(\alpha-2)(\alpha+2)}\Big]
\end{align}
and with minimum point $t_0=\theta$.
\end{itemize}
\end{theorem}
\begin{proof}
By Proposition \ref{Mode}, a mode $t$ of the LBS distribution satisfies the following cubic equation
\begin{align}\label{cubic-pol}
p_3(t)=t^3 - \theta (\alpha^2-1) t^2 + \theta^2 (\alpha^2-1) t -\theta^3 = 0.
\end{align}	
By using Descartes' rule of signs
(see, e.g. \cite{xue2000loop}) in \eqref{cubic-pol}, we have the following statements:  
\begin{itemize}
\item[a)] $p_3(t)$ has exactly one positive root when $\alpha\leqslant 1$;
\item[b)]  $p_3(t)$ has three or one positive roots when $\alpha>1$.	
\end{itemize}

On the other hand, it is well-known that
the discriminant of a cubic polynomial $ax^3 + bx^2 + cx + d$ is given by 
$\Delta =18abcd-4b^3d + b^2c^2- 4ac^3 - 27a^2d^2$.
In our case, we have
\begin{align*}
\Delta 
&=
\theta^6
\left[
(\alpha^2-1)^4-8(\alpha^2-1)^3+18(\alpha^2-1)^2-27
\right]
\\[0,2cm]
&=
\theta^6 \alpha^2 (\alpha-2)^3(\alpha+2)^3.
\end{align*}
Hence, the following statements with respect to $\Delta$ follow:
\begin{itemize}
	\item[c)] 
	$\Delta < 0$ when $\alpha< 2$. By using a) and b), this implies that
	$p_3(t)$ has exactly one positive root and two complex conjugate non-real roots;
	\item[d)] 
	$\Delta = 0$ when $\alpha= 2$.  By using a) and b), this implies that,
	$p_3(t)$ has exactly one positive triple root;
	\item[e)]  
	$\Delta >0$ when $\alpha>2$.  By using a) and b), this implies that,
	$p_3(t)$ has three distinct positive roots.	
\end{itemize}

We are now ready to prove Items 1) and 2). Indeed, note that $t_0=\theta$ is a critical point of the LBS density $f_T(t)$, $t>0$, because $p_3(\theta)=0$. Then,
since $\lim_{t\to 0^+}f_{T}(t)=0$ and $\lim_{t\to +\infty}f_{T}(t)=0$, Items c) and d) imply the statement in Item 1); and Item e) implies that the LBS density is bimodal as $\alpha> 2$ whenever $t_0=\theta$ is a minimum point. To complete the proof of Item 2), it remains to show that, when $\alpha>2$,
\begin{itemize}
\item[(I)]  $t_0=\theta$ is a minimum point for the LBS density, and that
\item[(II)] the modes of the LBS density are given by $t_{\pm}$ defined in \eqref{modes-t}.
\end{itemize}

To verify Item (I), it is sufficient to check that $f''_{Y}(\theta)>0$ when $\alpha> 2$.  Indeed, note that
\begin{align}\label{derivative-f-1}
f''_{Y}(t) = \phi\big[a(t)\big]
\left\{a'''(t)+\big[a^2(t)-1\big][a'(t)]^3-3a(t)a'(t)a''(t)\right\}
\end{align}
and that
\begin{align*}
f''_{T}(t) 
\stackrel{\eqref{n-derivative}}{=}
\dfrac{1}{\E(Y)} \,
\left[
2 f'_{Y}(t) + t f''_{Y}(t)
\right].
\end{align*}
Substituting \eqref{derivative-f-1} in the above equation we have
\begin{align*}
f''_T(t)&=
\dfrac{2}{\E(Y)} \,
f'_{Y}(t) 
+
\dfrac{t \phi\big[a(t)\big]}{\E(Y)} \, 
\left\{a'''(t)+\big[a^2(t)-1\big][a'(t)]^3-3a(t)a'(t)a''(t)\right\}.
\end{align*}
By using the relations $a(\theta)=0$, $f'_T(\theta)=0$, and the expressions for the derivatives of $a(t)$ given in \eqref{derivatives-a}, we get 
\begin{align*}
f''_T(\theta)
=
{(\alpha-2)(\alpha+2)\over 4\sqrt{2\pi}\alpha^3\theta^2 \E(Y)}
>0,
\end{align*}
whenever $\alpha>2$. Then the statement in Item (I) follows.

In what remains of the proof we check Item (II).
Indeed, since $t_0=\theta$ is a critical point for the LBS density, note the cubic equation \eqref{cubic-pol} can be written as
\begin{align}
p_3(t)=
(t-\theta)
\big[
t^2-\theta(\alpha^2-2)t+\theta^2
\big]
= 0.
\end{align}	
Bhaskara's Formula gives the following roots for $p_3(t)$:
\begin{align*}
t_0=\theta; \quad 
t_{\pm}=
{\theta\over 2}\,
\Big[(\alpha^2-2)\pm\alpha\sqrt{(\alpha-2)(\alpha+2)}\Big].
\end{align*} 
Since $t_{-}<t_0<t_{+}$ and, $\lim_{t\to 0^+}f_{T}(t)=0$ and $\lim_{t\to +\infty}f_{T}(t)=0$, the proof of Item (II) follows.
We thus complete the proof of Item 2).

\end{proof}

\subsubsection{Some properties of the HR of the LBS distribution}
The survival function (SF) and HR of the LBS distribution are given respectively by
\begin{equation}\label{eq:SF}
\begin{aligned}
S_T(t) &= 1-\int_{0}^{t}f_{T}(\xi)\, {\rm d}\xi , \quad t>0,
\\[0,2cm]
&=
1-\Phi\big[a(t)\big]
-
{\alpha^2\theta\over 2\E(Y)}\,
\big\{
{\rm e}^{2/\alpha^2}
\big\{
\Phi\left[ A(t) \right]-1
\big\}
-
\phi\big[a(t)\big]
\left[a(t)+ A(t) \right]
\big\},
\end{aligned}
\end{equation}
where $A(t)={\sqrt{4+\alpha^2a^2(t)}/\alpha}$, 
and
\begin{equation}\label{eq:HR}
\begin{aligned}
H_T(t)&={f_{T}(t)\over S_T(t)}, \quad t>0,
\\[0,2cm]
&=
{
{t\phi\big[a(t)\big] a'(t)}
\over
\E(Y)\big\{1-\Phi\big[a(t)\big]\big\}
-
{\alpha^2\theta\over 2}\,
\big\{
{\rm e}^{2/\alpha^2}
\big\{
\Phi\left[ A(t) \right]-1
\big\}
-
\phi\big[a(t)\big]
\left[a(t)+ A(t) \right]
\big\}
},
\end{aligned}
\end{equation}
with $\phi(\cdot)$ and $\Phi(\cdot)$ being the standard normal PDF and cumulative distribution function (CDF), respectively.

To enunciate and prove the following two results, we follow the same notations as Theorem \ref{Unimodality-bimodality}.
\begin{proposition}[Monotonicity of HR, case $\alpha\leqslant 2$]
	\label{Monotonicity-1}
	The HR of the LBS distribution with $\alpha\leqslant 2$ has the following monotonic properties:
\begin{itemize}
\item[1)] It is increasing for all $t<\theta$;
\item[2)] It is decreasing for all $t>t_1$, for some $t_1\geqslant\theta$.
\end{itemize}
\end{proposition}
\begin{proof}
When $\alpha\leqslant 2$ the LBS density is unimodal with mode $t_0=\theta$ (Theorem \ref{Unimodality-bimodality}). Then there is $t_1\geqslant\theta$ so that the LBS density is a concave upward function on the interval $I=(t_1,+\infty)$. In other words, $f_T''(t)>0$ for all $t\in I$. Equivalently, the function $f_T'(t)$ decreases on $I$. But since, by unimodality, the LBS density $f_T$ decreases on this interval, we have that the function $G_T(t)$, defined as
\begin{align*}
G_T(t)=-\, {f_T'(t)\over f_T(t)},
\end{align*}
decreases for all $t\in I$, because $G_T(t)$, $t\in I$, is a product of nonnegative
decreasing functions. Hence, by \cite{10.2307/2287666} the function $H_T(t)$ is decreasing for all $t>t_1$. This proves the second item.

In what follows we prove Item 1).
Indeed, by unimodality of the LBS density (Theorem \ref{Unimodality-bimodality}), the LBS dentity $f_T$ increases on $(0,\theta)$. Hence, since the SF $S_T(t)$, $t>0$, is decreasing, by definition of HR it follows that $H_T(t)$  is a product of nonnegative increasing functions.
Therefore, it is a increasing function for all $t<\theta$.

\end{proof}

\begin{proposition}[Monotonicity of HR, case $\alpha> 2$]
	The HR of the LBS distribution with $\alpha> 2$ has the following monotonic properties:
	\begin{itemize}
		\item[1)] It is increasing for all $t<t_{-}$ ou for all $\theta<t<t_{+}$;
		\item[2)] It is decreasing for all $t_1<t<\theta$ ou for all $t>t_2$, for some $t_{-}<t_1<\theta$ and $t_2>t_{+}$.
	\end{itemize}
\end{proposition}
\begin{proof}
The proof follows by using Theorem \ref{Unimodality-bimodality} and an analogous reasoning to the proof of Proposition \ref{Monotonicity-1}. Therefore, this one is omitted.

\end{proof}

\section{The LBS regression model}\label{sec:3}

In this section, we formulate the LBS regression model, and then detail the associated estimation, inference 
and residual analysis based on the maximum likehood method.

\subsection{The model and maximum likelihood estimation}\label{sec:3.1}

Let $T_1, T_2,\ldots, T_n$ denote independent random variables, where $T_i\sim \textrm{LBS}(\alpha_i, \theta_{i})$, with observed values denoted by $t_1, t_2, \ldots, t_n$, respectively. Then, the LBS regression model is formulated as 
\begin{equation}\label{eq:pred} 
\begin{gathered} 
\eta_{1i} = {g}_1(\theta_i) = \sum_{j=1}^{p}{ x_{ij}{\beta_j}} = \bm x_{i}^{\top}\bm{\beta}, \quad i=1,\ldots,n;
\\
\eta_{2i} = {g}_2(\alpha_i) = \sum_{j=1}^{q}{ w_{ij}{\rho_j}} = \bm w_{i}^{\top}\bm{\rho}, \quad i=1,\ldots,n;
\end{gathered}
\end{equation}
where $\bm{\beta} = (\beta_1, \dots, \beta_p)^{\top}$ and $\bm{\rho} = (\rho_1, \dots, \rho_q)^{\top}$ are vectors of unknown parameters, with $\bm{\beta} \in \mathbb{R}^{p}$, $\bm{\rho} \in \mathbb{R}^{q}$ and $p+q<n$. The vectors $\bm x_{i}^{\top} = (x_{i1}, \dots, x_{ip})$ and $\bm w_{i}^{\top} = (w_{i1}, \dots, w_{ip})$ are known covariates from the $i$-th row of the matrices $\mathbf{X} \in \mathbb{R}^{n \times p}$ and $\mathbf{W} \in \mathbb{R}^{n \times q}$ which are full rank, i.e., $rank(\mathbf{X}) = p$ and $rank(\mathbf{W}) = q$. Commonly $\mathbf{x}_{\mathord{\cdot} 1} = \bm1_{n}^{\top}$ and $\mathbf{w}_{\mathord{\cdot} 1} = \bm1_{n}^{\top}$, where $\bm1_{n}^{\top}$ is a 1's vector size $n$. The link function ${g:} (0, +\infty) \rightarrow \mathbb{R}$ is invertible and at least twice differentiable. Usually $g(\theta) = \log(\theta)$ (log function) and $g(\theta) = \sqrt{\theta}$ (square root function).

The corresponding likelihood function for $\bm\delta=(\bm\rho^\top, \bm\beta^\top)^\top$ is
\begin{equation}\label{lik:fun}
L(\bm\delta)=\prod_{i=1}^{n} f_T(t_i; \alpha_i, \theta_i),
\end{equation}
where $f(\cdot)$ is the LBS PDF given in \eqref{eq:fdp}. By taking the logarithm in \eqref{lik:fun}, we obtain the log-likelihood function for $\bm\delta=(\bm\rho^\top, \bm\beta^\top)^\top$ as
\begin{equation}
\label{eq:log-vero}
\ell(\bm\delta) \propto -\sum_{i=1}^{n}\frac{1}{2\alpha_i^2}\left(\frac{t_i}{\theta_i} + \frac{\theta_i}{t_i}-2\right) - \sum_{i=1}^{n}\log(2\alpha_i + \alpha_i^3) -\frac{3}{2}\sum_{i=1}^{n}\log(\theta_i) + \sum_{i=1}^{n}\log(t_i + \theta_i),
\end{equation}
where $\theta_{i} = {g_1}^{-1}(\eta_{1i})$ and $\alpha_{i} = {g_2}^{-1}(\eta_{2i})$, as defined in \eqref{eq:pred}. 

The maximum likelihood estimators of the LBS regression model parameters are the solution of the equation $\bm{\dot{\ell}} = \bm 0$, where $\bm{\dot\ell}=(\bm{\dot\ell}_{\beta}^\top, \bm{\dot\ell}_{\rho}^\top)^\top$ is the gradient vector, with the first derivatives given by
\begin{equation}
\bm{\dot{\ell}} = \begin{pmatrix}
\dfrac{\partial \ell(\bm\delta)}{\partial \bm\beta} \\
\\
\dfrac{\partial \ell(\bm\delta)}{\partial \bm\rho} \\
\end{pmatrix} =
\begin{pmatrix}
\bm{\dot{\ell}}_{\beta} \\ 
\bm{\dot{\ell}}_{\rho} \\
\end{pmatrix} =
\begin{pmatrix}
\bm{X}^\top \bm{\mathcal{A}} \bm{z} \\
\bm{W}^\top \bm{\mathcal{B}} \bm{c} \\
\end{pmatrix},
\end{equation}
with
\begin{equation}
\begin{aligned}
\label{eq:grad}
& \bm{\dot{\ell}}_{\beta} = \sum_{i=1}^{n} \underbrace{ \left\lbrace \frac{1}{t_i + \theta_i} - \frac{1}{2\alpha_i^2} \left(\frac{1}{t_i} - \frac{t_i}{\theta_i^2} \right) - \frac{3}{2\theta_i} \right\rbrace }_{z_i} \underbrace{ \frac{1}{g'_1(\theta_i)} }_{a_i} \bm{x}_{i} = \sum_{i=1}^{n} z_i a_i \bm{x}_{i}; \\
& \bm{\dot{\ell}}_{\rho} = \sum_{i=1}^{n} \underbrace{ \left\lbrace \frac{1}{\alpha_i^3} \left(\frac{t_i}{\theta_i} + \frac{\theta_i}{t_i}-2 \right) - \frac{(2 + 3\alpha_i^2)}{2\alpha_i + \alpha_i^3} \right\rbrace }_{c_i} \underbrace{ \frac{1}{g'_2(\alpha_i)} }_{b_i} \bm{w}_{i} = \sum_{i=1}^{n} c_i b_i \bm{w}_{i}; \\
\end{aligned}
\end{equation}
where $\bm{\dot{\ell}}_\beta = (\dot{\ell}_{\beta_1}, \ldots, \dot{\ell}_{\beta_p})^\top$, $\bm{\dot{\ell}}_\rho = (\dot{\ell}_{\rho_1}, \ldots, \dot{\ell}_{\rho_q})^\top$, $\bm{z} = (z_1,\ldots,z_n)^\top$, $\bm{c} = (c_1,\ldots,c_n)^\top$, $\bm{\mathcal{A}} = \text{diag}(a_1,\ldots,a_n)$ and $\bm{\mathcal{B}} = \text{diag}(b_1,\ldots,b_n)$. 

However, these equations do not have a closed form, requiring the use of iterative numerical methods to solve them. They are solved using the BFGS quasi-Newton method; see \cite{mjm:00}[p.\,199]. Under regularity conditions \citep{ch:74}, the asymptotic distribution of $\widehat{\bm{\delta}}$ is a multivariate normal, that is,
\begin{equation}\label{asym:dist}
 \sqrt{n}(\widehat{{\bm\delta}}-{\bm\delta})\dot{\sim}\textrm{N}_{p+q}\left(\bm{0}_{p+q}, {\bm{\Sigma}}_{{\bm{\delta}}}\right),
\end{equation}
where $\,\dot{\sim}\,$ denotes convergence in distribution and ${\bm{\Sigma}}_{{\bm{\delta}}}$ is the asymptotic covariance matrix of $\widehat{\bm{\delta}}$, which is the inverse of the expected Fisher information matrix. One can approximate the expected Fisher information matrix by its observed version obtained from the Hessian matrix $\ddot{\bm \ell}({\bm\delta})$, such that ${{\bm{\Sigma}}}_{{\bm{\delta}}}\approx [-\ddot{\bm \ell}({\bm\delta})]^{-1}$. The standard errors (SEs) can then be approximated by the square roots of the diagonal elements in the covariance matrix evaluated at $\widehat{{\bm\delta}}$. Note that
\begin{equation}
\bm{\ddot{\ell}}({\bm\delta}) = \begin{pmatrix}
\dfrac{\partial^2 \ell(\bm\delta)}{\partial \bm\beta \partial \bm\beta^\top} & \dfrac{\partial^2 \ell(\bm\delta)}{\partial \bm\beta \partial \rho^\top} \\ \\
\dfrac{\partial^2 \ell(\bm\delta)}{\partial \rho \partial \bm\beta^\top} & \dfrac{\partial^2 \ell(\bm\delta)}{\partial \rho \partial \rho^\top} \\
\end{pmatrix} =
\begin{pmatrix}
\bm{\ddot{\ell}}_{\beta\beta} & \bm{\ddot{\ell}}_{\beta\rho} \\ 
\bm{\ddot{\ell}}_{\rho\beta} & \bm{\ddot{\ell}}_{\rho\rho} \\
\end{pmatrix} =
\begin{pmatrix}
\bm{X}^\top \bm{\mathcal{V}} \bm{X} & \bm{X}^\top \bm{\mathcal{H}} \bm{W} \\
\bm{W}^\top \bm{\mathcal{H}} \bm{X}  & \bm{W}^\top \bm{\mathcal{U}} \bm{W} \\
\end{pmatrix},
\end{equation}
with
\begin{align*} \label{eq:hess}
& \bm{\ddot{\ell}}_{\beta\beta} = \sum_{i=1}^{n} \underbrace{ \left\lbrace \frac{3}{2\theta^2_i} - \frac{t_i}{\alpha_i^2 \theta^3_i} - \frac{1}{(t_i + \theta_i)^2}  \right\rbrace }_{z_i'} \underbrace{ \left[\frac{1}{g'_1(\theta_i)}\right]^2}_{a_i^2} \bm{x}_{i} \bm{x}_{i}^\top + \\ 
& \myquad[1.4] + \sum_{i=1}^{n} \underbrace{ \left\lbrace \frac{1}{t_i + \theta_i} - \frac{1}{2\alpha_i^2} \left(\frac{1}{t_i} - \frac{t_i}{\theta_i^2} \right) - \frac{3}{2\theta_i} \right\rbrace }_{z_i} \underbrace{ \left( -\frac{g''_1(\theta_i)}{[g'_1(\theta_i)]^2} \right) }_{d_i} \underbrace{ \frac{1}{g'_1(\theta_i)} }_{a_i} \bm{x}_{i} \bm{x}_{i}^\top \\
& \myquad[1.4] = \sum_{i=1}^{n} \underbrace{(z_{i}' a_i^2 + z_i d_i a_i)}_{v_{ii}} \bm{x}_{i} \bm{x}_{i}^\top = \sum_{i=1}^{n} v_{ii} \bm{x}_{i} \bm{x}_{i}^\top \\[10pt] 
&\bm{\ddot{\ell}}_{\rho \beta} = \bm{\ddot{\ell}}_{\beta \rho}^{\top} = \sum_{i=1}^{n} \underbrace{ \left\lbrace \frac{1}{\alpha_i^3} \left( \frac{1}{t_i} - \frac{t_i}{\theta^2_i} \right) \right\rbrace }_{k_i} \underbrace{ \frac{1}{g'_2(\alpha_i)} }_{b_i} \underbrace{ \frac{1}{g'_1(\theta_i)} }_{a_i} \bm{w}_{i} \bm{x}_{i}^\top  \stepcounter{equation}\tag{\theequation} \\
& \myquad[1.4] = \sum_{i=1}^{n} \underbrace{ k_i b_i a_i }_{h_{ii}} \bm{w}_{i} \bm{x}_{i}^\top = \sum_{i=1}^{n} h_{ii} \bm{w}_{i} \bm{x}_{i}^\top \\[10pt]
& \bm{\ddot{\ell}}_{\rho\rho} = \sum_{i=1}^{n} \underbrace{ \left\lbrace \frac{4+3\alpha_i^4}{(2\alpha_i+\alpha_i^3)^2} - \frac{3}{\alpha_i^4} \left( \frac{t_i}{\theta_i} + \frac{\theta_i}{t_i} - 2 \right) \right\rbrace }_{c'_i} \underbrace{ \left[\frac{1}{g'_2(\alpha_i)}\right]^2}_{b_i^2} \bm{w}_{i} \bm{w}_{i}^\top + \\ 
& \myquad[1.4] + \sum_{i=1}^{n} \underbrace{ \left\lbrace \frac{1}{\alpha_i^3} \left(\frac{t_i}{\theta_i} + \frac{\theta_i}{t_i}-2 \right) - \frac{(2 + 3\alpha_i^2)}{2\alpha_i + \alpha_i^3} \right\rbrace }_{c_i} \underbrace{ \left( -\frac{g''_2(\alpha_i)}{[g'_2(\alpha_i)]^2} \right) }_{e_i} \underbrace{ \frac{1}{g'_2(\alpha_i)} }_{b_i} \bm{w}_{i} \bm{w}_{i}^\top \\ 
& \myquad[1.4] = \sum_{i=1}^{n} \underbrace{(c_{i}' b_i^2 + c_i e_i b_i)}_{u_{ii}} \bm{w}_{i} \bm{w}_{i}^\top = \sum_{i=1}^{n} u_{ii} \bm{w}_{i} \bm{w}_{i}^\top \\ 
\end{align*}
where $\bm{\mathcal{Z'}} = \text{diag}(z'_1,\ldots,z'_n)$, $\bm{\mathcal{Z}} = \text{diag}(\bm{z})$, $\bm{\mathcal{C'}} = \text{diag}(c'_1,\ldots,c'_n)$, $\bm{\mathcal{C}} = \text{diag}(\bm{c})$, $\bm{\mathcal{D}} = \text{diag}(d_1,\ldots,d_n)$, $\bm{\mathcal{E}} = \text{diag}(e_1,\ldots,e_n)$, $\bm{\mathcal{K}} = \text{diag}(k_1,\ldots,k_n)$, $\bm{\mathcal{V}} = \bm{\mathcal{Z}}'\bm{\mathcal{A}}^2 + \bm{\mathcal{Z}}\bm{\mathcal{D}}\bm{\mathcal{A}}$, $\bm{\mathcal{H}} = \bm{\mathcal{K}}\bm{\mathcal{B}}\bm{\mathcal{A}}$ and $\bm{\mathcal{U}} = \bm{\mathcal{C}}'\bm{\mathcal{B}}^2 + \bm{\mathcal{C}}\bm{\mathcal{E}}\bm{\mathcal{B}}$. 


\subsection{Initial values}\label{sec:3.2}

The initial value for ${\bm{\beta}} = ({\beta}_1, \dots,{\beta}_p)^{\top}$ can be obtained by the use of least squares method from 
\begin{equation}
\label{eq:betas}
\widehat{\bm{\beta}}_0 = (\bm{X}^{\top}\bm{X})^{-1}\bm{X}^{\top}{g}_1(\bm{t}),
\end{equation}
where $\bm{t} = (t_{1}, \dots, t_{n})^{\top}$ and ${g}_1$ is the link function.

Let $\textstyle y_i = ({t_i\over \widehat{\theta}_i} + {\widehat{\theta}_i\over t_i} - 2)^{1/2}$ for $i=1,\ldots,n$. Then, the initial value for  $\bm{\rho} = (\rho_1, \dots, \rho_q)^{\top}$ can be estimated by ordinary least squares, as
\begin{equation}
\label{eq:rhos}
\widehat{\bm{\rho}}_0 = (\bm{W}^{\top}\bm{W})^{-1}\bm{W}^{\top}{g}_2(\bm{y}),
\end{equation}
where $\widehat{\theta_{i}} = {g}_1^{-1}(\bm{x}_{i}^{\top}\widehat{\bm{\beta}}_0)$, $\bm{y} = (y_{1}, \dots, y_{n})^{\top}$ and ${g}_2$ is the link function.

\subsection{Confidence intervals}\label{sec:3.3}

In this subsection, we derive some confidence intervals (CIs) using the asymptotic properties of maximum likelihood estimators and the bootstrap approach.

\subsubsection{Asymptotic confidence interval}\label{sec:3.3.1}

Based on the asymptotic normal approximation in \eqref{asym:dist}, we have the asymptotic CI (ACI) for the parameter $\delta_j$ as
\begin{equation}
\label{eq:ICnormal}
\big(
\widehat{\delta_j} - z_{1-\frac{\kappa}{2}}{{\bm{\Sigma}}_{{\bm{\delta}}}}_{jj},\,  \widehat{\delta_j} + z_{1-\frac{\kappa}{2}}{{\bm{\Sigma}}_{\bm{\delta}}}_{jj}
\big),
\end{equation}
where $z_{1-\frac{\kappa}{2}}$ is the ${1-\frac{\kappa}{2}}$ quantile of the standard normal distribution, and $j = 1, \dots, p + q$.

\subsubsection{Bootstrap confidence intervals}\label{sec:3.3.2}

The bootstrap approach, developed by \citet{efron1986bootstrap}, provides another way to obtain CIs. We compute the percentile bootstrap CI (PCI) and the bias-corrected and accelerated CI (BCI) for the parameter $\delta_j$, with $j = 1, \dots, p + q$, based on the the following steps: 
\begin{enumerate}[label=(\roman*)]
	\item Use the maximum likelihood method to obtain an estimate $\widehat{\bm\delta} = (\widehat{\bm{\rho}}^{\top}, \widehat{\bm{\beta}}^{\top})^\top$ of $\bm\delta = (\bm{\rho}^{\top}, \bm{\beta}^{\top})^\top$ from the original sample, and then estimate the components $(\widehat\alpha_{i}, \widehat\theta_{i})$, with $i = 1, \dots, n$. 
	
	\item For each bootstrap replica $b=1, \dots, B$:
	\begin{enumerate}
		\item For each $i = 1, \dots, n$, draw a pseudo-random sample $t^{*(b)}_{1}, \dots, t^{*(b)}_{n}$, where $t^{*(b)}_{i} \sim f(t ; \widehat\alpha_{i}, \widehat\theta_{i})$, as defined in \eqref{eq:fdp}.
		
		\item Compute the $b$-th replica of $\widehat{\bm\delta}^{(b)} = (\widehat{\bm{\rho}}^{(b)\top}, \widehat{\bm{\beta}}^{(b)\top})^\top$ by the maximum likelihood method based on the pseudo-sample generated in (a).
	\end{enumerate}  
\end{enumerate}    

\subsection{Residual analysis}\label{sec:3.4}

We perform residuals analysis in order to evaluate the validity of the assumptions of the model and also as tools for model selection. We consider three types of residuals. The first residual is the generalized Cox-Snell (GCS), given by
\begin{equation}
r_i^{GCS} = -\log(\widehat{S}_T(t_i)), \quad i=1,\dots,n,
\end{equation}
where $\widehat{S}_T(t_i) = 1 - \widehat{F}_T(t_i)$ is the survival function fitted to the data. The GCS residual is asymptotically standard exponential, $\textrm{EXP}(1)$ in short, when the model is correctly specified whatever the specification of the model is; 

The second residual is the randomized quantile (RQ), given by
\begin{equation}
r_i^{RQ} = \Phi^{-1}(\widehat{S}_T(t_i)), \quad i=1,\ldots,n,
\end{equation}
where $\Phi^{-1}$ is the inverse function of the  standard normal CDF and $\widehat{S}_T(t_i)$ is the survival function fitted to the data. The RQ residual follows a standard normal distribution when the model is specified correctly, regardless of the LBS model.

The third residual (U) is based on the relationship given by (P5) of Subsection \ref{sec:2.1}. Note that
\begin{equation}
\label{eq:mistura}
r_i^{U} = \frac{1}{{\alpha}_i^2}\biggl(\frac{t_i}{{\theta}_i}+\frac{{\theta}_i}{t_i}-2\biggr), \quad i=1,\ldots,n,
\end{equation}
should result in $r_i^{U}$ being independent and identically distributied observations from a mixture of two gamma distributions; see (P5). Thus, from the given data, we may estimate ${\theta}_i$ and ${\alpha}_i$ by $\widehat{\theta}_i = g_1^{-1}(\bm x_{i}^{\top}\widehat{\bm{\beta}})$ and
$\widehat{\alpha}_i = g_2^{-1}(\bm w_{i}^{\top}\widehat{\bm{\rho}})$, respectively, and use them to determine the $r_i^{U}$ values.

\section{Monte Carlo simulation studies}\label{sec:4}

Three Monte Carlo simulation studies are carried out to evaluate 
the performances of the maximum likelihood estimates, the coverage probabilities of the 95\% CIs and the empirical distribution of the residuals. We use
the \texttt{R} software to do all numerical calculations; see \cite{rmanual:20}. 

The simulation scenario considers sample size $n = 50, 100, 500$. For the $\theta$ component, the values of the true parameters are taken as $\beta_0 = 1$ and $ \beta_1 = -1$. For the $\alpha$ component, we consider the scenario with and without covariates. For the scenario without covariates, $\alpha = {0.25, 0.50, 1.00,2.00, 2.50}$ and for the scenario with covariates, $\rho_0 = -1$ and $\rho_1 = 0.25, 0.75, 1.25$. The covariates $x_{1i}$ and $w_ {1i}$ in the predictors of the \eqref{eq:pred} models were obtained from a uniform distribution in the interval $(-1, 1)$. We use 5,000 Monte Carlo replications for each combination of above given parameters and sample size; we use $R=500$ bootstrap replicates.

\subsection{Maximum likelihood estimates}

The maximum likelihood estimation results for the considered LBS regresion model are presented in Tables \ref{tab:1} and \ref{tab:2}. 
The empirical mean, bias and mean squared error (MSE) are reported. The results of Tables \ref{tab:1} and \ref{tab:2} allow us to conclude that, when the sample size increases, the empirical means tend to the reference true parameter values. Moreover, the empirical bias and MSE both decrease, as expected, when the sample size increases.

\begin{table}[!ht]
\centering
\small
	\caption{Empirical mean, bias and MSE from simulated data for the indicated maximum likelihood estimates of the LBS regression model parameters with covariates in $\alpha$.}
	\label{tab:1}
	\begin{tabular}{@{}lrrrrrrrrr@{}}
		\toprule
		\multirow{3}{*}{True value} & \multicolumn{3}{l}{Mean}    & \multicolumn{3}{l}{Bias}    & \multicolumn{3}{l}{MSE}  \\ \cmidrule(l){2-4} \cmidrule(l){5-7} \cmidrule(l){8-10} 
		& \multicolumn{3}{l}{$n$}       & \multicolumn{3}{l}{$n$}       & \multicolumn{3}{l}{$n$}    \\ \cmidrule(l){2-4} \cmidrule(l){5-7} \cmidrule(l){8-10}
		& \multicolumn{1}{l}{50} & \multicolumn{1}{l}{100} & \multicolumn{1}{l}{500} & \multicolumn{1}{l}{50} & \multicolumn{1}{l}{100} & \multicolumn{1}{l}{500} & \multicolumn{1}{l}{50} & \multicolumn{1}{l}{100} & \multicolumn{1}{l}{500} \\ \midrule
		$\beta_0=1$                          & 1.0049  & 1.0030  & 1.0005  & 0.0049  & 0.0030  & 0.0005  & 0.0032 & 0.0016 & 0.0003 \\
		$\beta_1=-1$                         & -1.0009 & -0.9991 & -1.0006 & -0.0009 & 0.0009  & -0.0006 & 0.0086 & 0.0038 & 0.0007 \\
		$\rho_0=-1$                         & -1.0438 & -1.0221 & -1.0042 & -0.0438 & -0.0221 & -0.0042 & 0.0128 & 0.0056 & 0.0010 \\
		$\rho_1=0.25$                       & 0.2691  & 0.2550  & 0.2517  & 0.0191  & 0.0050  & 0.0017  & 0.0280 & 0.0138 & 0.0025 \\ \midrule
		$\beta_0=1$                          & 1.0027  & 1.0020  & 1.0001  & 0.0027  & 0.0020  & 0.0001  & 0.0021 & 0.0011 & 0.0002 \\
		$\beta_1=-1$                         & -1.0011 & -0.9996 & -1.0006 & -0.0011 & 0.0004  & -0.0006 & 0.0056 & 0.0028 & 0.0005 \\
		$\rho_0=-1$                         & -1.0424 & -1.0215 & -1.0038 & -0.0424 & -0.0215 & -0.0038 & 0.0117 & 0.0052 & 0.0009 \\
		$\rho_1=0.75$                       & 0.7869  & 0.7639  & 0.7533  & 0.0369  & 0.0139  & 0.0033  & 0.0270 & 0.0129 & 0.0024 \\ \midrule
		$\beta_0=1$                          & 1.0012  & 1.0012  & 1.0001  & 0.0012  & 0.0012  & 0.0001  & 0.0010 & 0.0006 & 0.0001 \\
		$\beta_1=-1$                         & -1.0011 & -1.0000 & -1.0005 & -0.0011 & 0.0000  & -0.0005 & 0.0029 & 0.0016 & 0.0003 \\
		$\rho_0=-1$                         & -1.0413 & -1.0210 & -1.0038 & -0.0413 & -0.0210 & -0.0038 & 0.0105 & 0.0047 & 0.0008 \\
		$\rho_1=1.25$                       & 1.2955  & 1.2695  & 1.2543  & 0.0455  & 0.0195  & 0.0043  & 0.0254 & 0.0119 & 0.0022 \\ \bottomrule
	\end{tabular}
\end{table}

From Table \ref{tab:2}, we observe that that when $\alpha$ assumes values greater than 2, which is the bimodal case (Theorem 2.2), the MSE values are greater for $\beta_{0}$ e for $\rho_{0} = \ln(2) \ \text{or} \ \ln(2.5)$, but not for $\beta_{1}$. This result shows that when bimodality is present, the bias and MSE tend to be greater than the unimodal case. However, as the sample size increases, the bias and MSE decrease dramatically.

\begin{table}[!ht]
\centering
\small
	\caption{Empirical mean, bias and MSE from simulated data for the indicated maximum likelihood estimates of the LBS regression model parameters without covariates in $\alpha$.}
	\label{tab:2}
	\begin{tabular}{@{}lrrrrrrrrr@{}}
		\toprule
		\multirow{3}{*}{True value} & \multicolumn{3}{l}{Mean}    & \multicolumn{3}{l}{Bias}    & \multicolumn{3}{l}{MSE}  \\ \cmidrule(l){2-4} \cmidrule(l){5-7} \cmidrule(l){8-10} 
		& \multicolumn{3}{l}{$n$}       & \multicolumn{3}{l}{$n$}       & \multicolumn{3}{l}{$n$}    \\ \cmidrule(l){2-4} \cmidrule(l){5-7} \cmidrule(l){8-10}
		& \multicolumn{1}{l}{50} & \multicolumn{1}{l}{100} & \multicolumn{1}{l}{500} & \multicolumn{1}{l}{50} & \multicolumn{1}{l}{100} & \multicolumn{1}{l}{500} & \multicolumn{1}{l}{50} & \multicolumn{1}{l}{100} & \multicolumn{1}{l}{500} \\ \midrule
		$\beta_0=1$                            & 1.0028                 & 1.0017                  & 1.0003                  & 0.0028                 & 0.0017                  & 0.0003                  & 0.0014                 & 0.0007                  & 0.0001                  \\
		$\beta_1=-1$                          & -0.9992                & -0.9991                 & -1.0002                 & 0.0008                 & 0.0009                  & -0.0002                 & 0.0042                 & 0.0018                  & 0.0004                  \\
		$\rho_0=\ln(0.25)$                        & -1.4180                & -1.4035                 & -1.3904                 & -0.0317                & -0.0173                 & -0.0041                 & 0.0120                 & 0.0053                  & 0.0010                  \\ \midrule
		$\beta_0=1$                            & 1.0106                 & 1.0064                  & 1.0015                  & 0.0106                 & 0.0064                  & 0.0015                  & 0.0070                 & 0.0035                  & 0.0007                  \\
		$\beta_1=-1$                          & -0.9990                & -0.9983                 & -1.0003                 & 0.0010                 & 0.0017                  & -0.0003                 & 0.0158                 & 0.0067                  & 0.0013                  \\
		$\rho_0=\ln(0.50)$                        & -0.7263                & -0.7112                 & -0.6975                 & -0.0331                & -0.0181                 & -0.0044                 & 0.0130                 & 0.0058                  & 0.0011                  \\ \midrule
		$\beta_0=1$                            & 1.0316                 & 1.0233                  & 1.0058                  & 0.0316                 & 0.0233                  & 0.0058                  & 0.1290                 & 0.0274                  & 0.0057                  \\
		$\beta_1=-1$                          & -1.0010                & -0.9976                 & -1.0007                 & -0.0010                & 0.0024                  & -0.0007                 & 0.0429                 & 0.0180                  & 0.0036                  \\
		$\rho_0=\ln(1.00)$                        & -0.0389                & -0.0235                 & -0.0058                 & -0.0389                & -0.0235                 & -0.0058                 & 0.0430                 & 0.0113                  & 0.0023                  \\ \midrule
		$\beta_0=1$                            & 1.1071                 & 0.9653                  & 0.9975                  & 0.1071                 & -0.0347                 & -0.0025                 & 2.1341                 & 1.7120                  & 0.1965                  \\
		$\beta_1=-1$                          & -1.0034                & -0.9977                 & -1.0008                 & -0.0034                & 0.0023                  & -0.0008                 & 0.0551                 & 0.0230                  & 0.0046                  \\
		$\rho_0=\ln(2.00)$                        & 0.6240                 & 0.7026                  & 0.6925                  & -0.0692                & 0.0094                  & -0.0007                 & 0.5316                 & 0.4185                  & 0.0469                  \\ \midrule
		$\beta_0=1$                            & 1.3516                 & 1.0837                  & 0.9920                  & 0.3516                 & 0.0837                  & -0.0080                 & 2.8662                 & 2.4906                  & 0.3549                  \\
		$\beta_1=-1$                          & -1.0034                & -0.9976                 & -1.0008                 & -0.0034                & 0.0024                  & -0.0008                 & 0.0540                 & 0.0225                  & 0.0045                  \\
		$\rho_0=\ln(2.50)$                        & 0.7309                 & 0.8710                  & 0.9193                  & -0.1854                & -0.0453                 & 0.0030                  & 0.7141                 & 0.6094                  & 0.0840                  \\ \bottomrule
	\end{tabular}
\end{table}

\subsection{Coverage probabilities}

Tables \ref{tab:3} and \ref{tab:4} present the coverage probabilities of 95\% CIs presented Subsection \ref{sec:3.3} for the LBS regression model. The results show that the ACI, PCI and BCI coverage probabilities approach the nominal level of 95\% when the sample increases, as expected. However, when $\alpha$ takes values greater than 2, the ACI, PCI and BCI coverage probabilities are lower than the corresponding nominal value for $\beta_{0}=1$ and for $\rho_{0}=\ln(2) \ \text{or} \ \ln(2.5)$. In general, the BCI has the best performance, which might be due to its characteristics as this method corrects for bias and skewness in the distribution of bootstrap estimates.

\begin{table}[!ht]
\centering
\small
	\caption{Empirical coverage probabilities of 95\% CIS for the LBS regression model with covariates in $\alpha$.}
	\label{tab:3}
	\begin{tabular}{@{}lrrrrrrrrr@{}}
		\toprule
		\multirow{3}{*}{True value} & \multicolumn{3}{l}{ACI}                                                 & \multicolumn{3}{l}{PCI}                                              & \multicolumn{3}{l}{BCI}                                                    \\ \cmidrule(l){2-10} 
		& \multicolumn{3}{l}{$n$}       & \multicolumn{3}{l}{$n$}       & \multicolumn{3}{l}{$n$}    \\ \cmidrule(l){2-4} \cmidrule(l){5-7} \cmidrule(l){8-10}
		& \multicolumn{1}{l}{50} & \multicolumn{1}{l}{100} & \multicolumn{1}{l}{500} & \multicolumn{1}{l}{50} & \multicolumn{1}{l}{100} & \multicolumn{1}{l}{500} & \multicolumn{1}{l}{50} & \multicolumn{1}{l}{100} & \multicolumn{1}{l}{500} \\ \midrule
		$\beta_0=1$                           & 92.44                & 93.96                 & 94.70                 & 91.78                & 93.18                 & 94.06                 & 91.82                & 93.24                 & 93.96                 \\
		$\beta_1=-1$                          & 92.44                & 94.10                 & 95.04                 & 92.10                & 93.80                 & 94.54                 & 92.12                & 93.88                 & 94.48                 \\
		$\rho_0=-1$                          & 91.26                & 93.94                 & 95.16                 & 84.58                & 89.26                 & 93.96                 & 83.16                & 88.68                 & 93.60                 \\
		$\rho_1=0.25$                        & 93.12                & 94.50                 & 94.58                 & 92.78                & 94.12                 & 94.38                 & 92.80                & 93.98                 & 94.32                 \\ \midrule
		$\beta_0=1$                           & 92.24                & 93.78                 & 94.78                 & 91.66                & 93.08                 & 94.16                 & 91.68                & 93.06                 & 94.02                 \\
		$\beta_1=-1$                          & 91.60                & 94.12                 & 95.16                 & 91.32                & 93.78                 & 94.92                 & 91.30                & 93.82                 & 94.82                 \\
		$\rho_0=-1$                          & 91.36                & 93.46                 & 94.88                 & 83.86                & 89.34                 & 94.02                 & 82.26                & 88.74                 & 93.66                 \\
		$\rho_1=0.75$                        & 92.72                & 94.34                 & 94.48                 & 91.26                & 93.86                 & 94.14                 & 90.88                & 93.74                 & 93.82                 \\ \midrule
		$\beta_0=1$                           & 92.00                & 93.46                 & 94.76                 & 91.60                & 93.04                 & 94.10                 & 91.62                & 92.74                 & 94.12                 \\
		$\beta_1=-1$                          & 91.32                & 93.80                 & 95.08                 & 91.06                & 93.66                 & 94.72                 & 90.90                & 93.44                 & 94.54                 \\
		$\rho_0=-1$                          & 91.60                & 93.30                 & 94.76                 & 83.36                & 88.80                 & 93.78                 & 81.66                & 88.12                 & 93.70                 \\
		$\rho_1=1.25$                        & 92.68                & 94.02                 & 94.52                 & 89.50                & 92.50                 & 93.94                 & 89.04                & 92.56                 & 93.82                 \\ \bottomrule
	\end{tabular}
\end{table}

\begin{table}[!ht]
\centering
\small
	\caption{Empirical coverage probabilities of 95\% CIS for the LBS regression model without covariates in $\alpha$.}
	\label{tab:4}
	\begin{tabular}{@{}lrrrrrrrrr@{}}
		\toprule
		\multirow{3}{*}{True value} & \multicolumn{3}{l}{ACI}                                                 & \multicolumn{3}{l}{PCI}                                              & \multicolumn{3}{l}{BCI}                                                    \\ \cmidrule(l){2-10} 
		& \multicolumn{3}{l}{$n$}       & \multicolumn{3}{l}{$n$}       & \multicolumn{3}{l}{$n$}    \\ \cmidrule(l){2-4} \cmidrule(l){5-7} \cmidrule(l){8-10}
		& \multicolumn{1}{l}{50} & \multicolumn{1}{l}{100} & \multicolumn{1}{l}{500} & \multicolumn{1}{l}{50} & \multicolumn{1}{l}{100} & \multicolumn{1}{l}{500} & \multicolumn{1}{l}{50} & \multicolumn{1}{l}{100} & \multicolumn{1}{l}{500} \\ \midrule
		$\beta_0=1$                           & 93.96                & 94.98                 & 94.84                 & 93.22                & 94.44                 & 94.56                 & 93.28                & 94.56                 & 94.60                 \\
		$\beta_1=-1$                          & 93.66                & 94.88                 & 95.34                 & 93.32                & 94.42                 & 95.08                 & 93.28                & 94.40                 & 95.04                 \\
		$\rho_0=\ln(0.25)$                        & 92.74                & 94.44                 & 94.46                 & 88.72                & 92.12                 & 94.06                 & 87.80                & 91.54                 & 93.80                 \\ \midrule
		$\beta_0=1$                           & 93.36                & 94.28                 & 94.62                 & 92.26                & 93.36                 & 94.10                 & 92.32                & 93.48                 & 94.12                 \\
		$\beta_1=-1$                          & 93.70                & 94.86                 & 95.40                 & 93.20                & 94.60                 & 95.04                 & 93.28                & 94.34                 & 94.88                 \\
		$\rho_0=\ln(0.50)$                        & 92.66                & 94.42                 & 94.44                 & 88.28                & 91.70                 & 93.68                 & 87.34                & 91.16                 & 93.42                 \\ \midrule
		$\beta_0=1$                           & 90.32                & 92.62                 & 94.44                 & 88.58                & 91.88                 & 93.68                 & 88.18                & 92.14                 & 93.66                 \\
		$\beta_1=-1$                          & 94.04                & 94.80                 & 95.08                 & 93.76                & 94.46                 & 94.88                 & 93.64                & 94.44                 & 94.62                 \\
		$\rho_0=\ln(1.00)$                        & 90.58                & 93.10                 & 94.22                 & 86.46                & 90.46                 & 93.44                 & 85.48                & 90.32                 & 93.36                 \\ \midrule
		$\beta_0=1$                           & 75.32                & 81.40                 & 90.18                 & 77.62                & 84.06                 & 91.12                 & 81.98                & 87.08                 & 91.46                 \\
		$\beta_1=-1$                          & 94.24                & 94.78                 & 94.90                 & 93.88                & 94.50                 & 94.58                 & 93.86                & 94.38                 & 94.44                 \\
		$\rho_0=\ln(2.00)$                        & 75.62                & 81.60                 & 90.22                 & 77.20                & 83.42                 & 90.72                 & 81.00                & 86.08                 & 91.08                 \\ \midrule
		$\beta_0=1$                           & 66.58                & 73.58                 & 87.70                 & 74.38                & 82.10                 & 91.38                 & 81.42                & 87.44                 & 93.14                 \\
		$\beta_1=-1$                          & 94.40                & 94.94                 & 95.04                 & 94.14                & 94.66                 & 94.68                 & 94.00                & 94.54                 & 94.60                 \\
		$\rho_0=\ln(2.50)$                        & 67.32                & 73.72                 & 87.72                 & 74.14                & 81.90                 & 91.40                 & 80.38                & 86.98                 & 93.02                 \\ \bottomrule
	\end{tabular}
\end{table}

\subsection{Empirical distribution of residuals}

We now present the Monte Carlo simulation results for evaluating the performance of the $r^{\textrm{GCS}}$, $r^{\textrm{RQ}}$ and $r^{\textrm{U}}$ residuals. 
Tables \ref{tab:5} and \ref{tab:6} presents the empirical mean, standard deviation (SD), coefficient of skewness (CS) and coefficient of kurtosis (CK), whose values are expected to be as in Table \ref{tab:residuals:ref}, for the $r^{\textrm{GCS}}$, $r^{\textrm{RQ}}$ and $r^{\textrm{U}}$ residuals.  
From Tables \ref{tab:5} and \ref{tab:6}, we note that as the sample size increases, the values of the empirical mean, SD, CS and CK approach these values
of the reference distributions shown in Table \ref{tab:residuals:ref}. Therefore, the considered residuals conform well with the reference distributions.

\begin{table}[!ht]
\centering
\small
	\caption{Measures of the $r^{\textrm{GCS}}$, $r^{\textrm{RQ}}$ and $r^{\textrm{U}}$ residuals.}
	\label{tab:residuals:ref}
	\begin{tabular}{@{}lccc@{}}
		\toprule
		Measure         & $r^{\textrm{GCS}}$ & $r^{\textrm{RQ}}$ & $r^{\textrm{U}}$ \\ \midrule  
		Mean     &   1       &    0     & $3-\dfrac{4}{\alpha^2+2}$ \\[0.7cm]
		SD 		 &   1       &    1     & $\sqrt{6-\dfrac{16}{(\alpha^2+2)^2}}$ \\[0.7cm]
		CS 		 &   2       &    0     & $\dfrac{8(3\alpha^6 + 18\alpha^4 + 36\alpha^2 + 8)}{(6\alpha^4 + 24\alpha^2 + 8)^{3/2}}$ \\[0.7cm]
		CK       &   9       &    3     & $\dfrac{12(21\alpha^8 + 168\alpha^6 + 456\alpha^4 + 480\alpha^2 + 80)}{(6\alpha^4 + 24\alpha^2 + 8)^{2}}$   \\[0.3cm] \bottomrule
	\end{tabular}
\end{table}

%
%

\begin{table}[!ht]
\centering
\small
	\caption{Summary statistics for the $r^{\textrm{GCS}}$, $r^{\textrm{RQ}}$ and $r^{\textrm{U}}$ residuals with covariates in $\alpha$ ($\beta_0=1$, $\beta_1=-1$, $\rho_0=-1$).}
	\label{tab:5}
	\begin{tabular}{@{}lrrrrrrrrr@{}}
		\toprule
		\multirow{3}{*}{Statistic} & \multicolumn{3}{l}{$r^{\textrm{GCS}}$}                                 & \multicolumn{3}{l}{$r^{\textrm{RQ}}$}                                  & \multicolumn{3}{l}{$r^U$}                                  \\ \cmidrule(l){2-4} \cmidrule(l){5-7} \cmidrule(l){8-10}
		& \multicolumn{3}{l}{$n$}                                     & \multicolumn{3}{l}{$n$}                                      & \multicolumn{3}{l}{$n$}                                      \\ \cmidrule(l){2-4} \cmidrule(l){5-7} \cmidrule(l){8-10}
		& \multicolumn{1}{l}{50} & \multicolumn{1}{l}{100} & \multicolumn{1}{l}{500} & \multicolumn{1}{l}{50} & \multicolumn{1}{l}{100} & \multicolumn{1}{l}{500} & \multicolumn{1}{l}{50} & \multicolumn{1}{l}{100} & \multicolumn{1}{l}{500} \\ \midrule
		& \multicolumn{9}{l}{$\rho_1=0.25$} \\ \midrule
		Mean                       & 0.9999                 & 1.0000                  & 0.9999                  & -0.0011                & -0.0002                 & -0.0001                 & 1.1256                 & 1.1270                  & 1.1307                  \\
		SD                         & 0.9775                 & 0.9904                  & 0.9969                  & 1.0092                 & 1.0048                  & 1.0009                  & 1.4916                 & 1.5369                  & 1.5728                  \\
		CS                         & 1.5475                 & 1.7426                  & 1.9310                  & 0.0378                 & 0.0123                  & 0.0058                  & 2.0560                 & 2.3326                  & 2.6393                  \\
		CK                         & 5.6464                 & 6.8645                  & 8.3488                  & 2.7956                 & 2.8948                  & 2.9781                  & 7.7884                 & 9.8548                  & 12.8408                 \\ \midrule
		& \multicolumn{9}{l}{$\rho_1=0.75$}                                                                                                                                                                                                             \\ \midrule
		Mean                       & 1.0012                 & 1.0006                  & 1.0004                  & -0.0035                & -0.0013                 & -0.0008                 & 1.1640                 & 1.1602                  & 1.1659                  \\
		SD                         & 0.9771                 & 0.9904                  & 0.9970                  & 1.0064                 & 1.0037                  & 1.0005                  & 1.5345                 & 1.5755                  & 1.6134                  \\
		CS                         & 1.5528                 & 1.7473                  & 1.9316                  & 0.0324                 & 0.0101                  & 0.0052                  & 2.0458                 & 2.3172                  & 2.6102                  \\
		CK                         & 5.6805                 & 6.8971                  & 8.3551                  & 2.7995                 & 2.8997                  & 2.9790                  & 7.7461                 & 9.7546                  & 12.5849                 \\ \midrule
		& \multicolumn{9}{l}{$\rho_1=1.25$}                                                                                                                                                                                                             \\ \midrule
		Mean                       & 1.0019                 & 1.0010                  & 1.0004                  & -0.0051                & -0.0020                 & -0.0008                 & 1.2384                 & 1.2286                  & 1.2359                  \\
		SD                         & 0.9770                 & 0.9906                  & 0.9970                  & 1.0040                 & 1.0026                  & 1.0004                  & 1.6149                 & 1.6505                  & 1.6894                  \\
		CS                         & 1.5618                 & 1.7542                  & 1.9334                  & 0.0275                 & 0.0074                  & 0.0048                  & 2.0176                 & 2.2776                  & 2.5478                  \\
		CK                         & 5.7339                 & 6.9420                  & 8.3706                  & 2.8089                 & 2.9064                  & 2.9803                  & 7.6149                 & 9.5076                  & 12.0720                 \\ \bottomrule
	\end{tabular}
\end{table}

\begin{table}[!ht]
\centering
\small
	\caption{Summary statistics for the $r^{\textrm{GCS}}$, $r^{\textrm{RQ}}$ and $r^{\textrm{U}}$ residuals without covariates in $\alpha$ ($\beta_0=1$, $\beta_1=-1$).}
	\label{tab:6}
	\begin{tabular}{@{}lrrrrrrrrr@{}}
		\toprule
		\multirow{3}{*}{Statistic} & \multicolumn{3}{l}{$r^{\textrm{GCS}}$}                                 & \multicolumn{3}{l}{$r^{\textrm{RQ}}$}                                  & \multicolumn{3}{l}{$r^U$}                                  \\ \cmidrule(l){2-4} \cmidrule(l){5-7} \cmidrule(l){8-10}
		& \multicolumn{3}{l}{$n$}                                     & \multicolumn{3}{l}{$n$}                                      & \multicolumn{3}{l}{$n$}                                      \\ \cmidrule(l){2-4} \cmidrule(l){5-7} \cmidrule(l){8-10}
		& \multicolumn{1}{l}{50} & \multicolumn{1}{l}{100} & \multicolumn{1}{l}{500} & \multicolumn{1}{l}{50} & \multicolumn{1}{l}{100} & \multicolumn{1}{l}{500} & \multicolumn{1}{l}{50} & \multicolumn{1}{l}{100} & \multicolumn{1}{l}{500} \\ \midrule
		& \multicolumn{9}{l}{$\rho_0=\ln(0.25)$}\\ \midrule
		Mean                       & 1.0013                 & 1.0011                  & 1.0005 & -0.0013                & -0.0012                 & -0.0006 & 1.0581                 & 1.0591                  & 1.0602  \\
		SD                         & 0.9964                 & 0.9975                  & 0.9994 & 1.0099                 & 1.0048                  & 1.0009  & 1.4512                 & 1.4703                  & 1.4893  \\
		CS                         & 1.6539                 & 1.7915                  & 1.9457 & -0.0011                & -0.0004                 & -0.0006 & 2.1920                 & 2.4287                  & 2.7033  \\
		CK                         & 6.1768                 & 7.1628                  & 8.4571 & 2.8832                 & 2.9342                  & 2.9822  & 8.5821                 & 10.5504                 & 13.4125 \\ \midrule
		& \multicolumn{9}{l}{$\rho_0=\ln(0.50)$}                                                                                                                                                            \\ \midrule
		Mean                       & 1.0003                 & 1.0002                  & 1.0000 & 0.0001                 & 0.0001                  & 0.0000  & 1.2130                 & 1.2168                  & 1.2208  \\
		SD                         & 0.9946                 & 0.9963                  & 0.9989 & 1.0102                 & 1.0051                  & 1.0010  & 1.6337                 & 1.6556                  & 1.6774  \\
		CS                         & 1.6453                 & 1.7866                  & 1.9453 & 0.0038                 & 0.0024                  & 0.0004  & 2.1283                 & 2.3365                  & 2.5793  \\
		CK                         & 6.1364                 & 7.1363                  & 8.4548 & 2.8744                 & 2.9293                  & 2.9813  & 8.2144                 & 9.9125                  & 12.3395 \\ \midrule
		& \multicolumn{9}{l}{$\rho_0=\ln(1.00)$}                                                                                                                                                            \\ \midrule
		Mean                       & 0.9998                 & 0.9999                  & 1.0000 & 0.0017                 & 0.0008                  & 0.0002  & 1.6375                 & 1.6492                  & 1.6622  \\
		SD                         & 0.9897                 & 0.9931                  & 0.9981 & 1.0113                 & 1.0056                  & 1.0011  & 1.9901                 & 2.0168                  & 2.0452  \\
		CS                         & 1.6113                 & 1.7636                  & 1.9391 & 0.0239                 & 0.0151                  & 0.0036  & 1.8374                 & 1.9877                  & 2.1610  \\
		CK                         & 5.9778                 & 7.0118                  & 8.4133 & 2.8428                 & 2.9130                  & 2.9777  & 6.7649                 & 7.8887                  & 9.3977  \\ \midrule
		& \multicolumn{9}{l}{$\rho_0=\ln(2.00)$}                                                                                                                                                            \\ \midrule
		Mean                       & 0.9977                 & 0.9986                  & 0.9997 & 0.0126                 & 0.0077                  & 0.0018  & 2.1833                 & 2.2628                  & 2.3194  \\
		SD                         & 0.9832                 & 0.9895                  & 0.9972 & 1.0274                 & 1.0158                  & 1.0036  & 2.2428                 & 2.2937                  & 2.3419  \\
		CS                         & 1.5768                 & 1.7413                  & 1.9332 & 0.0953                 & 0.0576                  & 0.0132  & 1.5351                 & 1.6346                  & 1.7486  \\
		CK                         & 5.8354                 & 6.9030                  & 8.3785 & 2.8490                 & 2.9029                  & 2.9698  & 5.5880                 & 6.3225                  & 7.2565  \\ \midrule
		& \multicolumn{9}{l}{$\rho_0=\ln(2.50)$}                                                                                                                                                            \\ \midrule
		Mean                       & 0.9972                 & 0.9982                  & 0.9996 & 0.0106                 & 0.0060                  & 0.0017  & 2.2528                 & 2.3701                  & 2.4905  \\
		SD                         & 0.9751                 & 0.9830                  & 0.9959 & 1.0262                 & 1.0147                  & 1.0039  & 2.2443                 & 2.3079                  & 2.3808  \\
		CS                         & 1.5659                 & 1.7324                  & 1.9311 & 0.1321                 & 0.0872                  & 0.0224  & 1.5002                 & 1.5872                  & 1.6827  \\
		CK                         & 5.8005                 & 6.8689                  & 8.3697 & 2.9153                 & 2.9578                  & 2.9848  & 5.4820                 & 6.1538                  & 6.9842  \\ \bottomrule
	\end{tabular}
\end{table}

\clearpage

\section{Application to real data}\label{sec:5}

In this section, the LBS regression model is illustrated using data from the Meteorological Database for Teaching and Research (BDMEP) for the years 2011-2016, from the Brazilian National Institute of Meteorology (INMET) (\href{https://bdmep.inmet.gov.br/}{Source: INMET Network Data}). The dependent variable ($t_i$) is water evaporation $(\textit{mm})$. The covariates considered in the study were: $x_{i1}$ is the actual evapotranspiration (\textit{mm}); $x_{i2}$ is the total insolation (\textit{h}); $x_{i3}$ is the cloudiness (tenths); and $x_{i4}$ is the relative humidity (\%). The evaporation is measured by the piche evaporimeter, all monthly averages, observed at a monitoring station located in Bras\'ilia, Brazil. Other environmental variables were considered in an initial analysis, however only the afore-mentioned variables were considered statistically significant, at the 5\% level of significance, thus remaining in the final model adopted in this application.

Regarding the dependent variable, water evaporation, the actual evapotranspiration, cloudiness and relative humidity variables have a negative correlation of -0.55 and -0.74, -0.97, respectively, while the total insolation variable has a positive correlation (0.77). Therefore, in general, the lower the levels of actual evapotranspiration, cloudiness and relative humidity in the environment, the greater the water evaporation. On the other hand, the greater the total insolation, the greater the amount of water evaporation. These results are in line with what was expected in the environment.

{Table~\ref{tab:summary} reports descriptive statistics of the observed water evaporation, including the minimum, median, mean, maximum, SD, coefficient of variation (CV), CS and CK values. From this table, we observe a skewed and high 
kurtosis features in the data.}

\begin{table}[!ht]
	\centering
	\small
	\caption{Summary statistics for the water evaporation data.}
	\label{tab:summary}
	\begin{tabular}{lrrrrrrrrr}
	\toprule
	Variable	&	$n$	&	Min. &	Median	&	Mean	&	Max.	&	SD	&	CV (\%)	&	CS	&	CK \\ 
	\midrule
	Water evaporation	&	70	&	65.3	&	138.55	&	156.88	&	303.80	&	66.68	&	42.50	&	0.73	&	-0.67 \\
	\bottomrule
	\end{tabular}
\end{table}

Figure~\ref{fig:density_boxplot} presents an estimated density superimposed on the histogram and boxplots for the water evaporation data. The adjusted boxplot for the water evaporation data indicates that some outliers are not identified by the usual boxplot; see Figure~\ref{fig:density_boxplot}(right). The adjusted boxplot is used when the data is skew distributed; see \cite{hvv:08}.  Note that the skewness observed in Table~\ref{tab:summary} is confirmed by the histogram presented in Figure~\ref{fig:density_boxplot}(left); this figure also indicates bimodality. Thus, the LBS regression model seems to be appropriate to describe these data.

\begin{figure}[!ht]
	\centering
	\includegraphics[scale=0.8]{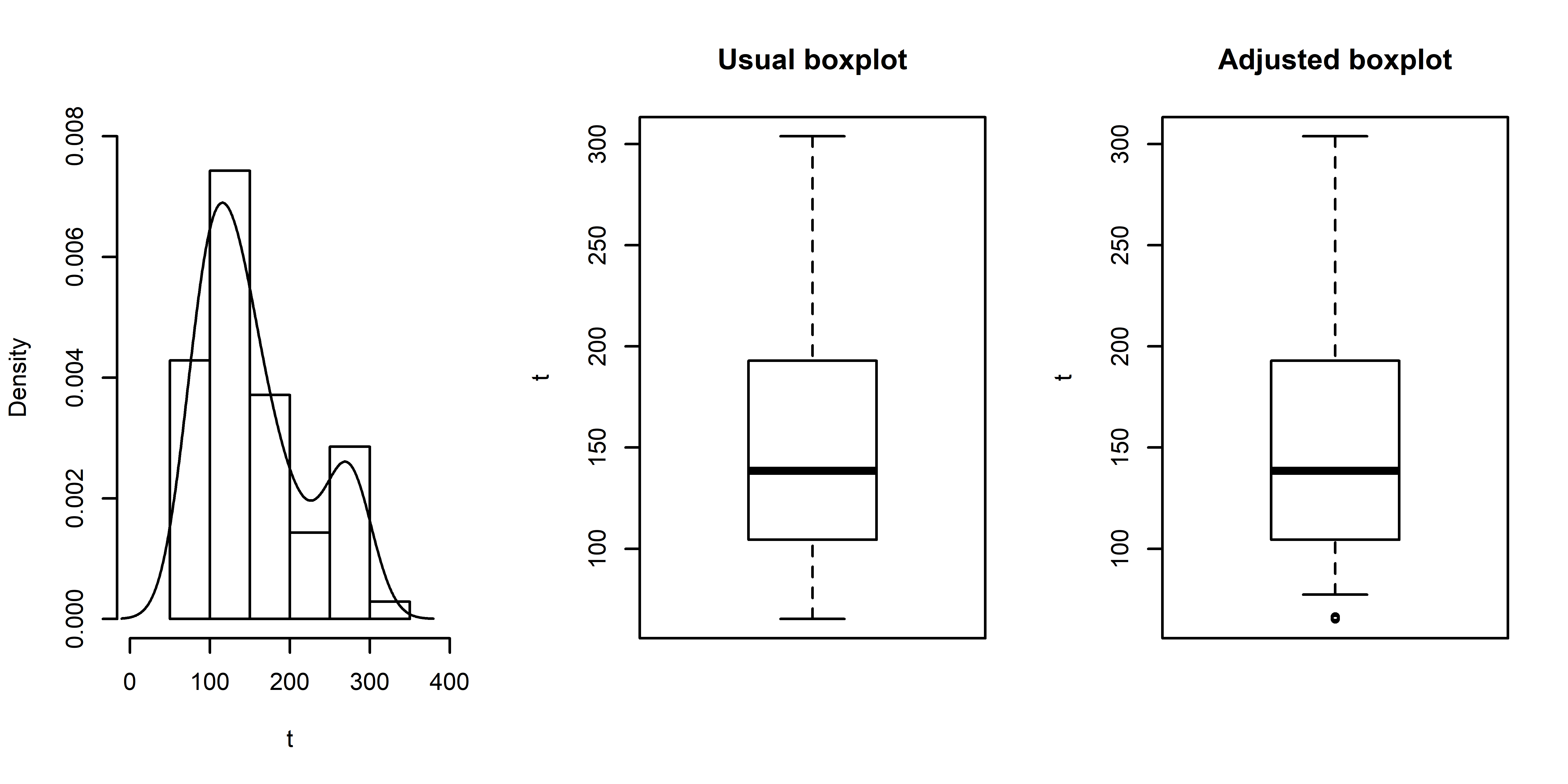}
	\caption{Estimated density superimposed on the histogram and boxplots for the water evaporation data.}\label{fig:density_boxplot}
\end{figure}

We analyze the water evaporation data using the LBS regression model, expressed as
\begin{equation}\label{eq:model}
\begin{aligned}
&\ln(\theta_i) = \beta_0 + \beta_1 x_{1i} + \beta_2 x_{2i} + \beta_3 x_{3i} + \beta_4 x_{4i},\\
&\ln(\alpha_i) = \rho_0 + \rho_1 x_{2i} + \rho_2 x_{3i}.
\end{aligned}
\end{equation}

Table \ref{tab:MLE_CI} reports the maximum likelihood estimates, computed by the BFGS quasi-Newton method, 
SEs and 95\% CI estimates. Note that the 95\% CIs do not include the null value (except the intercept for the $\alpha$ component), then the coefficients are statistically significant.

\begin{table}[!ht]
	\centering
	\small
	\caption{Point and interval estimates for the LBS regression model for the water evaporation data.}
	\label{tab:MLE_CI}
	\begin{tabular}{@{}lrcccc@{}}
	\toprule
	& Point Estimation  && \multicolumn{3}{c}{Interval estimation}     \\ \cline{2-2} \cline{4-6}  
		& Estimates    && ACI        & PCI & BCI                \\ \midrule
		$\theta$ components    &         &                    &                      &                    \\
		(Intercept) & 6.7430 (0.1723)  && (6.4052, 7.0808)   & (6.4009, 7.0722)     & (6.3994, 7.0692)   \\
		Actual evapotranspiration ($x_{1i}$) & 0.0015 (0.0004)  && (0.0007, 0.0023)   & (0.0008, 0.0023)     & (0.0007, 0.0023)   \\
		Total insolation ($x_{2i}$) & 0.0011 (0.0005)  && (0.0002, 0.0020)   & (0.0002, 0.0020)     & (0.0002, 0.0020)   \\
		Cloudiness ($x_{3i}$) & 0.0434 (0.0152)  && (0.0136, 0.0732)   & (0.0166, 0.0742)     & (0.0178, 0.0744)   \\
		Relative humidity ($x_{4i}$) & -0.0366 (0.0010)  && (-0.0386, -0.0346) & (-0.0386, -0.0345)   & (-0.0387, -0.0346) \\ \midrule
		$\alpha$ components    &         &                    &                      &                    \\
		(Intercept) & 1.0396 (1.4433)  && (-1.7892, 3.8683)  & (-2.7158, 4.8980)    & (-2.8676, 4.7398)  \\
		Total insolation  ($x_{2i}$) & -0.0130 (0.0042)  &&(-0.0212, -0.0048) & (-0.0250, -0.0028)   & (-0.0247, -0.0026) \\
		Cloudiness ($x_{3i}$) & -0.2324 (0.1117)  && (-0.4513, -0.0135) & (-0.5124, 0.0430)    & (-0.5076, 0.0465)  \\ \bottomrule
	\end{tabular}
\end{table}

Figure~\ref{fig:envelopes} displays the quantile versus quantile (QQ) plots with simulated envelope
of the $r^{\textrm{GCS}}$, $r^{\textrm{RQ}}$ and $r^{\textrm{U}}$ residuals for the LBS regression model. This figure indicates that these residuals in the LBS regression model show good agreements with the expected distributions. Figure \ref{fig:res_pred} plots  the residuals against the predicted values. Note that this figure shows random patterns, indicating a good fit for the LBS regression model. In addition, the Ljung-Box test results for up to 4th and 16th order serial correlations provide no evidence of serial correlation in the raw residuals ($t_i-\widehat{t}_i$, $i=1,\ldots,n$), with $p$-values equal to $0.6843$ and $0.4048$, respectively.

\begin{figure}[!ht]
	\centering
	\includegraphics[scale=0.8]{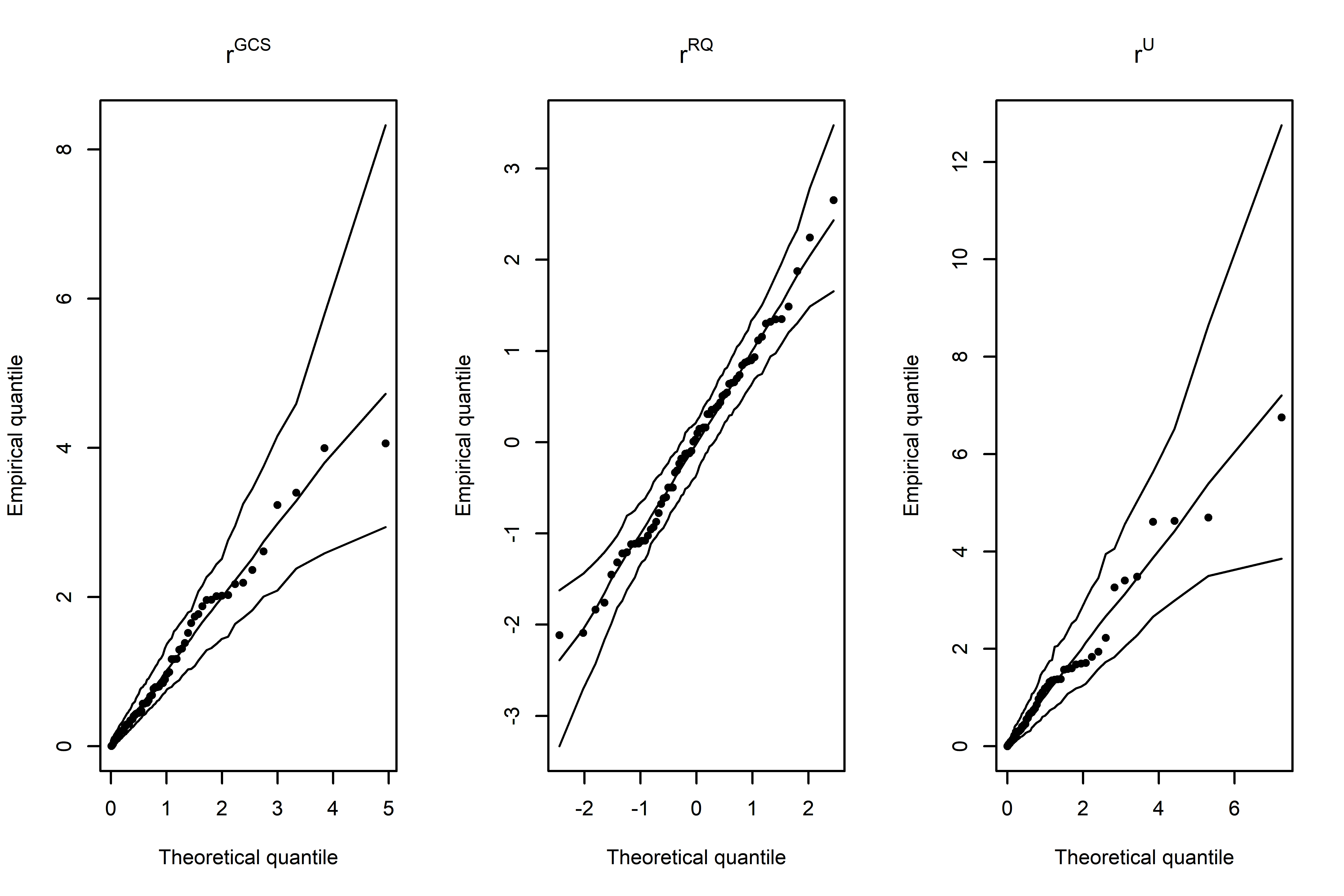}
	\caption{QQ plot and its envelope for the $r^{\textrm{GCS}}$, $r^{\textrm{RQ}}$ and $r^{\textrm{U}}$ residuals for the LBS regression model for the water evaporation data.}\label{fig:envelopes}
\end{figure}

\begin{figure}[!ht]
	\centering
	\includegraphics[scale=0.8]{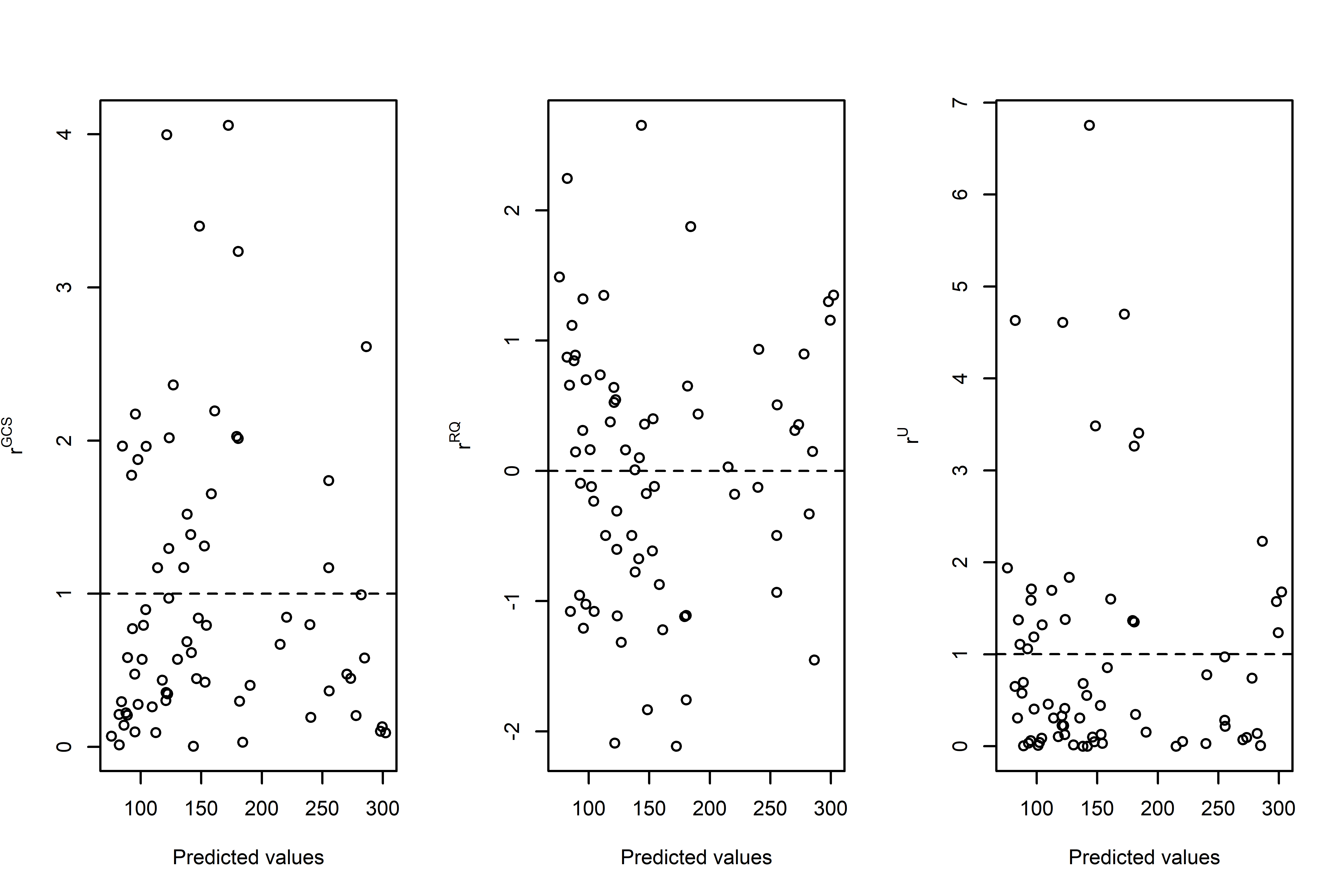}
	\caption{Predicted values against the $r^{\textrm{GCS}}$, $r^{\textrm{RQ}}$ and $r^{\textrm{U}}$ residuals.}\label{fig:res_pred}
\end{figure}

\section{Concluding remarks}\label{sec:6}

We have proposed and analyzed a new regression model based on the length-biased version of the Birnbaum-Saunders distribution proposed by \cite{lsa:09}. We have derived novel properties of the length-biased Birnbaum-Saunders distribution, which is both useful and practical for environmental sciences. We have considered the maximum likelihood method for parameter estimation. We have addressed interval estimation and studied three types of residuals. Monte Carlo simulations were carried out to evaluate the behaviour of the maximum likelihood estimates, the coverage probabilities of the confidence intervals and the empirical distribution of the residuals. The simulation results (a) have shown good performaces of the maximum likelihood estimates; (b) indicated that the bias-corrected and accelerated confidence interval  has the performance; and (c) indicated that the considered residuals conform well with their reference distributions. We have applied the proposed length-biased Birnbaum-Saunders regression model to a real meteorological data. The application has favored 
the use of the proposed regression model. As part of future research, it will be of interest to implement influence diagnostic tools. Furthermore, multivariate versions of the proposed length-biased Birnbaum-Saunders regression model can be studied. Finally, generalization the proposed model for the case with censored data can be investigated. Work on these problems is currently in progress and we hope to report these findings in future papers.

\section*{Acknowledgments}

We gratefully acknowledge financial support from CAPES and CNPq, Brazil.


\begin{thebibliography}{}

\bibitem[Balakrishnan and Kundu, 2019]{bk:19}
Balakrishnan, N. and Kundu, D. (2019).
\newblock {B}irnbaum-{S}aunders distribution: {A} review of models, analysis,
  and applications.
\newblock {\em Applied Stochastic Models in Business and Industry}, 35:4--49.

\bibitem[Balakrishnan and Zhu, 2015]{bz:15}
Balakrishnan, N. and Zhu, X. (2015).
\newblock Inference for the {B}irnbaum-{S}aunders lifetime regression model
  with applications.
\newblock {\em Communications in Statistics-Simulation and Computation},
  44(8):2073--2100.

\bibitem[Birnbaum and Saunders, 1969]{bs:69}
Birnbaum, Z.~W. and Saunders, S.~C. (1969).
\newblock A new family of life distributions.
\newblock {\em Journal of Applied probability}, 6(2):319--327.

\bibitem[Cox and Hinkley, 1974]{ch:74}
Cox, D.~R. and Hinkley, D.~V. (1974).
\newblock {\em {Theoretical Statistics}}.
\newblock Chapman and Hall, London, UK.

\bibitem[Dasilva et~al., 2020]{ddlms:20}
Dasilva, A., Dias, R., Leiva, V., Marchant, C., and Saulo, H. (2020).
\newblock [{I}nvited tutorial] {B}irnbaum-{S}aunders regression models: a
  comparative evaluation of three approaches.
\newblock {\em Journal of Statistical Computation and Simulation},
  90(14):2552--2570.

\bibitem[Efron and Tibshirani, 1986]{efron1986bootstrap}
Efron, B. and Tibshirani, R. (1986).
\newblock Bootstrap methods for standard errors, confidence intervals, and
  other measures of statistical accuracy.
\newblock {\em Statistical science}, pages 54--75.

\bibitem[Glaser, 1980]{10.2307/2287666}
Glaser, R.~E. (1980).
\newblock Bathtub and related failure rate characterizations.
\newblock {\em Journal of the American Statistical Association},
  75(371):667--672.

\bibitem[Hubert and Vandervieren, 2008]{hvv:08}
Hubert, M. and Vandervieren, E. (2008).
\newblock {An adjusted boxplot for skewed distributions}.
\newblock {\em Computational Statistics and Data Analysis}, 52:5186--5201.

\bibitem[Leiva et~al., 2020]{lsgs:2020c}
Leiva, V., S\'anchez, L., Galea, M., and Saulo, H. (2020).
\newblock Global and local diagnostic analytics for a geostatistical model
  based on a new approach to quantile regression.
\newblock {\em Stochastic Environmental Research and Risk Assessment},
  34:1457--1471.

\bibitem[Leiva et~al., 2009]{lsa:09}
Leiva, V., Sanhueza, A., and Angulo, J.~M. (2009).
\newblock A length-biased version of the birnbaum--saunders distribution with
  application in water quality.
\newblock {\em Stochastic Environmental Research and Risk Assessment},
  23(3):299--307.

\bibitem[Leiva et~al., 2014]{lscb:14}
Leiva, V., Santos-Neto, M., Cysneiros, F. J.~A., and Barros, M. (2014).
\newblock {B}irnbaum-{S}aunders statistical modelling: a new approach.
\newblock {\em Statistical Modelling}, 14:21--48.

\bibitem[Mittelhammer et~al., 2000]{mjm:00}
Mittelhammer, R.~C., Judge, G.~G., and Miller, D.~J. (2000).
\newblock {\em {Econometric Foundations}}.
\newblock Cambridge University Press, New York.

\bibitem[Patil, 2006]{p:06}
Patil, G. (2006).
\newblock {\em Weighted Distributions}.
\newblock American Cancer Society.

\bibitem[{R Core Team}, 2020]{rmanual:20}
{R Core Team} (2020).
\newblock {\em R: A Language and Environment for Statistical Computing}.
\newblock R Foundation for Statistical Computing, Vienna, Austria.

\bibitem[Rieck and Nedelman, 1991]{rn:91}
Rieck, J.~R. and Nedelman, J.~R. (1991).
\newblock A log-linear model for the {B}irnbaum-{S}aunders distribution.
\newblock {\em Technometrics}, 33(1):51--60.

\bibitem[S\'anchez et~al., 2020a]{slgs:2020a}
S\'anchez, L., Leiva, V., Galea, M., and Saulo, H. (2020a).
\newblock Birnbaum-saunders quantile regression and its diagnostics with
  application to economic data.
\newblock {\em Applied Stochastic Models in Business and Industry}, page pages
  in press available at http://doi.org/10.1002/asmb.2556.

\bibitem[S\'anchez et~al., 2020b]{slgs:2020b}
S\'anchez, L., Leiva, V., Galea, M., and Saulo, H. (2020b).
\newblock Birnbaum-saunders quantile regression models with application to
  spatial data.
\newblock {\em Mathematics}, 8:1000.

\bibitem[Sansgiry and Akman, 2001]{sa:01}
Sansgiry, P.~S. and Akman, O. (2001).
\newblock Reliability estimation via length-biased transformation.
\newblock {\em Communications in Statistics - Theory and Methods},
  30(11):2473--2479.

\bibitem[Xue, 2000]{xue2000loop}
Xue, J. (2000).
\newblock {\em Loop Tiling for Parallelism}.
\newblock The Springer International Series in Engineering and Computer
  Science. Springer US.

\end{thebibliography}
\end{document}